\documentclass[11pt]{article}
\usepackage[T1]{fontenc}
\usepackage{amsmath,amssymb,amsthm,mathtools,geometry,booktabs}
\usepackage[unicode]{hyperref}
\usepackage{xcolor}
\newcommand{\HH}{\mathrm H}
\newcommand{\II}{\mathrm I}

\newcommand{\ZZ}{\mathbb Z}

\newcommand{\supp}{\operatorname{supp}}
\newcommand{\unif}{\mathrm{unif}}

\newcommand{\cG}{\mathcal G}

\newtheorem{theorem}{Theorem}
\newtheorem{lemma}[theorem]{Lemma}
\newtheorem{proposition}[theorem]{Proposition}
\newtheorem{corollary}[theorem]{Corollary}
\newtheorem{definition}{Definition}
\theoremstyle{remark}
\newtheorem{remark}[theorem]{Remark}

\numberwithin{theorem}{section}
\numberwithin{definition}{section}
\usepackage{authblk}

\title{From an entropic inequality for sums in abelian groups\\
to the theorems of Ruzsa and Pl\"unnecke}

\author[1]{Chandra Nair}
\author[2]{Aref Namayandeh}
\author[3]{Zhaobang Zhu}

\affil[1,2,3]{Department of Information Engineering, CUHK, Hong Kong}

\date{\today}

\date{\today}

\begin{document}
\maketitle

\begin{abstract}
We give a direct information-theoretic proof of the
Pl\"unnecke--Ruzsa inequalities for finite subsets of abelian groups.
The argument rests on an elementary coupling inequality proved by
maximizing joint entropy, without tensor powers or asymptotic
constructions. It also yields Ruzsa's sumset triangle inequality
and an extension to non-abelian groups.
\end{abstract}

\section{Introduction}\label{sec:intro}

Inequalities for sumsets and inequalities for entropies measure the same
thing --- growth under addition --- in two different languages. Ruzsa
\cite{ruzsa09} proposed a ``theory of translation'' between them, and by
now there is a substantial body of work developing this dictionary: the
entropic Ruzsa calculus of Tao \cite{tao10}, the partition-determined
inequalities of Madiman, Marcus and Tetali \cite{mmt12}, the
differential-entropy sumset inequalities of Kontoyiannis and Madiman\cite{km012}, the
recent refinements by Green, Manners and Tao \cite{gmt24}, and the exact equivalence
theorems of Ruzsa \cite{ruzsa09} and of Lau and Nair \cite{ln25}, the
latter enlarging the translatable families via Sanov's theorem and
identifying maximal couplings of sums as the natural entropic counterpart
of sumset cardinalities. This contrasts with many earlier approaches, in which sums of independent random variables served as analogues of sumset inequalities. One main motivation of the research is that inequalities for independent random variables often do not imply their additive-combinatorial counterparts.

We refer to the survey \cite{nair26} for the state
of this programme and for references to the large body of work that is not cited here.

Gaps remain in the connections between these fields. 
For example, cornerstone cardinality inequalities, including Ruzsa’s sumset triangle inequality
\[
|A|\,|B+C|\leq |A+B|\,|A+C|,
\] 
Pl\"unnecke's inequality \cite{plunnecke70} and Ruzsa's generalization \cite{ruzsa09} on magnification ratios,
originally proved using graph-theoretic methods and later reproved by Petridis
\cite{petridis12}, were translatable into entropic formulation only through the
equivalence theorems, whose typical-set arguments are asymptotic and
require a torsion-free ambient group. In 2026, Lau, Nair, and Zhu \cite{lnz26} used ideas from convex optimization to give a direct information-theoretic proof of a maximal coupling form of Shearer's lemma for sums of distributions, which implies a corresponding inequality in additive combinatorics. In the same paper, a maximal-coupling entropic analogue of the displayed inequality was explicitly
conjectured for general abelian groups.

We prove the conjectured entropic analogue by establishing a stronger inequality. For any joint distribution \(p_{X,Y,Z}\) with marginals \(p_X,p_Y,p_Z\),
\begin{equation}
\HH_p(X)+\HH_p(X+Y+Z)
\le
\max_{q\in\Pi(p_X,p_Y)}\HH_q(X+Y)
+
\max_{r\in\Pi(p_X,p_Z)}\HH_r(X+Z).
\label{eq:motiv}
\end{equation}
Using this maximal-coupling inequality, we formulate and prove the entropic analogues of the Pl\"unnecke--Ruzsa inequality for magnification ratios. Together with the entropic characterization of the magnification ratio in \cite{ln25}, these yield, in a few lines, Ruzsa's sumset triangle inequality (Theorem~\ref{thm:ruzsa}), and the
different-summand Pl\"unnecke--Ruzsa inequality
(Theorem~\ref{thm:ruzsa-lemma}), whose specialization is Pl\"unnecke's inequality (Theorem~\ref{thm:plunnecke}). These results make the connections between the entropic and additive-combinatorial statements explicit. We also establish a sumset analogue of the displayed maximal-coupling inequality (Proposition~\ref{prop:abc}). Finally, carrying out the same argument with the ordering \(U=Y\circ X\circ Z\) yields a new sumset inequality for arbitrary, possibly non-abelian groups (Corollary~\ref{cor:nonabelian-combinatorial}), recovering
Ruzsa's sumset triangle inequality in the abelian case.

\emph{Outline.} Subsection~\ref{sec:prelim} introduces the notation and elementary
facts used throughout the paper.
Section~\ref{sec:main} states and proves the main coupling inequality and its entropic consequences. Section~\ref{sec:ac} transfers them to sumset inequalities. Section~\ref{sec:nonab} gives the non-abelian extension.

\subsection{Notation and elementary facts} \label{sec:prelim}
Unless otherwise stated, $G$ is an abelian group. We write $+$ for the operation in abelian groups and $\circ$ for the operation in general groups. The sets $A,B,C,D\subseteq G$ are finite
and nonempty, and $A+B:=\{a+b:a\in A,\ b\in B\}$ is the sumset of $A$ and $B$.
We write $\supp p_X$ for the support of $p_X$. For distributions $p_X$ and $p_Y$ on finite
alphabets, let $\Pi(p_X,p_Y)$ denote the convex and compact set of
couplings, i.e., joint distributions of $(X,Y)$ with $X\sim p_X$ and
$Y\sim p_Y$. The magnification ratio of $(A,D)$ is defined by
$\mu(A,D):=\min_{\emptyset\ne S\subseteq A}|S+D|/|S|$. For given $\supp p_X\subseteq A$, define
$\Phi_D(p_X):=
\max_{{W_{Y\mid X}}}
\bigl[\HH(X+Y)-\HH(X)\bigr]$,
where we consider $p_X,W_{Y|X}$ as an input distribution and a channel respectively, with the induced marginal $\supp p_Y \subseteq D$.

Throughout, we use the concavity of $\HH(X)$ in $p_X$, the bounds
$0\le\HH(X)\le\log|\supp p_X|$, the chain rule, subadditivity
$\HH(X_1+X_2)\le\HH(X_1,X_2)\le\HH(X_1)+\HH(X_2)$, non-negativity of
conditional mutual information, the fact that conditioning cannot
increase entropy, and invariance of entropy under bijections of finite
alphabets.


\section{The main coupling inequality}\label{sec:main}
In this section, we prove our main coupling inequality and use it to derive other non-trivial entropic inequalities. The next theorem strengthens \eqref{eq:motiv} and directly yields our desired max-coupling sum inequality.

\begin{theorem}[max-coupling entropy inequality]\label{thm:three}
Let $p_{X,Y,Z}$ be a distribution on $A \times B \times C$ and let $U=X+Y+Z$. 
Define 
\[
\mathcal{M}(p_X,p_Y,p_Z,p_U) := \{q_{X,Y,Z}:\ q_X=p_X,\ q_Y=p_Y,\ q_Z=p_Z,\ q_U=p_U\},
\]
the set of couplings with marginals $p_X,p_Y,p_Z$ and sum distribution $p_U$. Then

\begin{subequations}
\begin{align}
\HH(X)+\HH_p(X+Y+Z) &\le \max_{q \in \mathcal{M}} \big(\HH(X+Y)+\HH(X+Z)\big) \label{eq:three-max-sum}\\
&\le \max_{\Pi(p_X,p_Y, p_Z)} \big(\HH(X+Y)+\HH(X+Z)\big) \label{eq:three-coupling} \\ 
&\le \max_{\Pi(p_X,p_Y)}\HH(X+Y)+\max_{\Pi(p_X,p_Z)}\HH(X+Z) \label{eq:three}.
\end{align}
\end{subequations}

\end{theorem}

\begin{proof}
Fix the coupling $p_{X,Y,Z}\in\Pi(p_X,p_Y,p_Z)$ and set $U=X+Y+Z$, $R=X+Y$, $S=X+Z$; let $p_U$
be the induced distribution of $U$. The set $\mathcal M$ is a nonempty, convex and compact subset of
$\Pi(p_X,p_Y,p_Z)$. Since $\HH_q(X,Y,Z)$ is strictly concave in $q$, it has a unique maximizer $q^\star\in\mathcal M$.

Construct a distribution on $(X,R,U)$ by the Markov recombination
$q^\star_R\,q^\star_{X\mid R}\,q^\star_{U\mid R}$, and let
$q^{(R)}_{X,Y,Z}$ be the joint distribution it induces via $Y=R-X$, $Z=U-R$. The
pairwise marginals satisfy
\[
q^\star_{X,R}=q^{(R)}_{X,R},\qquad q^\star_{R,U}=q^{(R)}_{R,U};
\]
since $Y=R-X$ and $Z=U-R$ are functions of these pairs, all one-dimensional marginals of $X,Y,Z,U$ match, so
$q^{(R)}_{X,Y,Z}\in\mathcal M$. Moreover, $(X,Y,Z)$ and $(X,R,U)$ are in bijection, so
\[
\HH_{q^{(R)}}(X,Y,Z)=\HH_{q^{(R)}}(X,R,U)
=\HH_{q^{(R)}}(R)+\HH_{q^{(R)}}(X\mid R)+\HH_{q^{(R)}}(U\mid R).
\]
As the pairwise marginals agree, this rewrites as
\[
\begin{aligned}
\HH_{q^{(R)}}(X,Y,Z)
&=\HH_{q^\star}(R)+\HH_{q^\star}(X\mid R)
+\HH_{q^\star}(U\mid R)\\
&=\HH_{q^\star}(X,Y,Z)+\II_{q^\star}(X;U\mid R).
\end{aligned}
\]
Maximality of $\HH_{q^\star}(X,Y,Z)$ and non-negativity of $\II(X;U\mid R)$ forces $\II_{q^\star}(X;U\mid R)=0$.

By the symmetric Markov recombination on $(X,S,U)$, also
$\II_{q^\star}(X;U\mid S)=0$. In the abelian setting the tuples $(X,Y,Z)$, $(X,R,U)$, $(X,S,U)$,
$(R,X,S)$, $(R,U,S)$ are in bijection with one another and hence share the same joint entropy. Therefore, for any joint distribution $q_{X,Y,Z}$, chain rule yields
\[
\HH_q(R)+\HH_q(S)-\HH_q(X)-\HH_q(U)
=\II_q(R;S\mid X)+\II_q(R;S\mid U)-\II_q(X;U\mid R)-\II_q(X;U\mid S).
\]
Evaluating the above equality at $q^\star$, the last two terms vanish, so
\[
\HH_{q^\star}(R)+\HH_{q^\star}(S)-\HH_{q^\star}(X)-\HH_{q^\star}(U)
=\II_{q^\star}(R;S\mid X)+\II_{q^\star}(R;S\mid U)\ \ge\ 0.
\]
Since $q^\star_X=p_X$, $q^\star_U=p_U$, and $q^{\star} \in \mathcal{M}$, we conclude
\[
\HH(X)+\HH(X+Y+Z)\ \le\
\max_{q \in \mathcal{M}} \big(\HH(X+Y)+\HH(X+Z)\big),
\]
which is \eqref{eq:three-max-sum}. Inequality
\eqref{eq:three-coupling} follows from
$\mathcal M\subseteq\Pi(p_X,p_Y,p_Z)$, and \eqref{eq:three} follows
by bounding each term by its corresponding pairwise maximum.
\end{proof}

\begin{remark}\label{rem:abelianimp}
The second recombination uses commutativity: it preserves the distributions of
$(X,X+Z)$ and of $(X+Z,X+Y+Z)$, and in a non-abelian group this need not
preserve the distribution of $Y$. The correct non-abelian ordering is given in
Section~\ref{sec:nonab}.
\end{remark}



\begin{corollary}[general max-coupling entropy inequality]
\label{cor:k-coupling}
Let $k\ge1$ and let
$(X,Y_1,\dots,Y_k)\sim p$ be a distribution on
$A\times\prod_{i=1}^k B_i$, with marginals
$p_X,p_{Y_1},\dots,p_{Y_k}$. Then
\[
\HH_p(X+Y_1+\cdots+Y_k)-\HH(X)
\le
\sum_{i=1}^k
\left[
\max_{\Pi(p_X,p_{Y_i})}\HH(X+Y_i)
-\HH(X)
\right].
\]
\end{corollary}

\begin{proof}
Induct on $k$, the case $k=1$ being immediate. For $k\ge2$, put
$T=Y_2+\cdots+Y_k$. By \eqref{eq:three},
\[
\HH_p(X+Y_1+T)-\HH(X)
\le
\max_{\Pi(p_X,p_{Y_1})}\HH(X+Y_1)
+\max_{\Pi(p_X,p_T)}\HH(X+T)-2\HH(X).
\]
Every coupling of $p_X,p_T$ extends via
$p_{Y_2,\dots,Y_k\mid T}$ to a coupling with the prescribed
marginals and $T=Y_2+\cdots+Y_k$. Applying the induction
hypothesis to each such extension bounds the second maximum by
\[
\HH(X)+\sum_{i=2}^k
\left[\max_{\Pi(p_X,p_{Y_i})}\HH(X+Y_i)-\HH(X)\right],
\]
which gives the claim.
\end{proof}

\begin{corollary}\label{cor:lnz}
Let $p_X,p_Y,p_Z$ be distributions on $A,B,C$, respectively, and let $(X,Y,Z)\sim p$ be any coupling with these marginals. Then
\begin{equation}\label{eq:two}
\HH(X)+\HH_p(Y+Z)\ \le\
\max_{\Pi(p_X,p_Y)}\HH(X+Y)+\max_{\Pi(p_X,p_Z)}\HH(X+Z).
\end{equation}
Consequently,
\begin{equation}\label{eq:two-max}
\HH(X)+\max_{\Pi(p_Y,p_Z)}\HH(Y+Z)\ \le\
\max_{\Pi(p_X,p_Y)}\HH(X+Y)+\max_{\Pi(p_X,p_Z)}\HH(X+Z).
\end{equation}

\end{corollary}

\begin{proof}
Note that if $W$ is independent of $X$, then
\begin{equation}\label{eq:conv}
\HH(X+W)\ge\HH(X+W\mid X)=\HH(W\mid X)=\HH(W).
\end{equation}
Take a coupling of $p_Y,p_Z$ maximizing $\HH(Y+Z)$ and extend it by an independent $X\sim p_X$. Applying \eqref{eq:conv} with $W=Y+Z$
and then \eqref{eq:three} gives \eqref{eq:two-max}. Since $\HH_p(Y+Z)$ is bounded by the corresponding coupling maximum, \eqref{eq:two} follows.
\end{proof}

\begin{remark}
Inequality \eqref{eq:two-max} was conjectured by Lau, Nair and Zhu \cite{lnz26} for distributions supported on finite subsets of abelian groups. We prove it here as a direct consequence of Theorem~\ref{thm:three}. For finitely generated torsion-free abelian groups, it was known via the equivalence theorem \cite{ln25} (see Corollary 2 and Remark 4 therein), but no direct proof avoiding the equivalence theorem was previously known.
\end{remark}

\begin{corollary}[entropic Ruzsa distance triangle inequality]
\label{cor:d} Let $p_X, p_Y, p_Z$ be supported on finite subsets of an abelian group $G$. Then
$d(p_X,p_Z)\le d(p_X,p_Y)+d(p_Y,p_Z)$, where
\[
d(p_X,p_Y):=\max_{\Pi(p_X,p_Y)}\HH(X+Y)-\tfrac12\HH(X)-\tfrac12\HH(Y).
\]
\end{corollary}

\begin{proof}
Apply Corollary~\ref{cor:lnz} with the three distributions cyclically permuted and
rearrange.
\end{proof}
For a distribution $p_X$ such that  $\supp(p_X) \subseteq A$ and $Y$ taking values in $D$, recall
\[\begin{aligned}
\Phi_{D}(p_X)&= \max_{W_{Y\mid X}} \HH(X+Y) - \HH(X).
\end{aligned}
\]

\begin{theorem}[entropic Pl\"unnecke--Ruzsa inequality]\label{thm:fi-pr}
For every $p_X$ on $A$ and all finite nonempty $B,C\subseteq G$,
\begin{equation}\label{eq:fi-pr}
\Phi_{B+C}(p_X)\ \le\ \Phi_B(p_X)+\Phi_C(p_X).
\end{equation}
\end{theorem}

\begin{proof}
Let $(X,W)\sim p_{X,W}$ be any coupling with $X\sim p_X$ and $\supp p_{W}\subseteq B+C$. Choose a representation $W=Y+Z$, with $\supp Y\subseteq B$ and $\supp Z\subseteq C$, as a function of $W$, and denote the induced marginals by $p_Y,p_Z$. By \eqref{eq:three},
\[
\begin{aligned}
\HH(X)+\HH_p(X+W)
&\le
\max_{\Pi(p_X,p_Y)}\HH_q(X+Y)
+\max_{\Pi(p_X,p_Z)}\HH_q(X+Z)\\
&\le 2\HH(X)+\Phi_B(p_X)+\Phi_C(p_X),
\end{aligned}
\]
since $\supp p_Y\subseteq B$ and $\supp p_Z\subseteq C$. Subtracting $2\HH(X)$ and
maximizing over all such couplings of $(X,W)$, with a fixed $p_X$, gives \eqref{eq:fi-pr}.
\end{proof}

\begin{corollary}[entropic Pl\"unnecke--Ruzsa inequality for $k$ summands]\label{cor:k-fi-pr}
For every $k \ge 1$, finite nonempty $B_1,\dots,B_k\subseteq G$ and every $\supp p_X \subseteq A$,
\[
\Phi_{B_1+\cdots+B_k}(p_X)
\le \sum_{i=1}^k\Phi_{B_i}(p_X).
\]
Consequently,
\[
\min_{\substack{p_X\\\supp p_X\subseteq A}}\Phi_{B_1+\cdots+B_k}(p_X)
\le
\min_{\substack{p_X\\\supp p_X\subseteq A}}\sum_{i=1}^k\Phi_{B_i}(p_X).
\]
\end{corollary}

\begin{proof}
The first claim follows by induction from Theorem~\ref{thm:fi-pr},
applied to $(B_1+\cdots+B_{k-1},B_k)$.
The second follows by minimizing over $\supp p_X\subseteq A$.
\end{proof}

These inequalities yield the combinatorial results in
Section~\ref{sec:ac} through the following entropic characterization of the magnification ratio, due to Lau and Nair.

\begin{theorem}[entropic characterization of magnification ratio \cite{ln25}]\label{thm:magnif}
Let $A,D\subseteq G$ be finite nonempty sets. Then
\begin{equation}\label{eq:magnif}
\min_{\substack{p_X\\\supp p_X\subseteq A}}\Phi_D(p_X)
=\log\mu(A,D)
=\log\min_{\emptyset\ne S\subseteq A}\frac{|S+D|}{|S|}.
\end{equation}
\end{theorem}

The original proof of the lower bound in \cite{ln25} (Theorem 7) analyzes an optimal channel $W^\star_{Y|X}$ for each input distribution
$p_X$ using the KKT conditions, then reweights the input classes of
an optimal $p_X^\star$ while preserving the optimality of $W^\star_{Y|X}$.
For completeness, we provide an alternative proof using the
max-flow min-cut theorem in Appendix~\ref{apdx:magnif}.

\section{Implications for additive combinatorics}\label{sec:ac}
We now use the entropic inequalities from the previous section to obtain
standalone proofs of well-known inequalities in additive combinatorics,
as well as inequalities that are new to us. The aim is to enrich the
connection between inequalities in the two fields and contribute to
the rich programme pioneered by Ruzsa \cite{ruzsa09} and Tao \cite{tao10}. 

\subsection{Ruzsa's triangle inequality for sumsets}\label{sec:ruzsa}

\begin{theorem}[Ruzsa's sumset triangle inequality]\label{thm:ruzsa}
For all finite nonempty $A,B,C\subseteq G$,
\begin{equation}\label{eq:ruzsa}
|A||B+C|\ \le\ |A+B||A+C|.
\end{equation}
\end{theorem}

\begin{proof}
Choose a joint distribution $p_{Y,Z}$ on $B\times C$ such that
$Y+Z$ is uniform on $B+C$, and let $p_Y,p_Z$ be its marginals.
Put $p_X=\unif_A$. By \eqref{eq:two-max},
\[
\HH(X)+\max_{\Pi(p_Y,p_Z)}\HH(Y+Z)
\le
\max_{\Pi(p_X,p_Y)}\HH(X+Y)
+\max_{\Pi(p_X,p_Z)}\HH(X+Z).
\]
Here $\HH(X)=\log|A|$ and
$\max_{\Pi(p_Y,p_Z)}\HH(Y+Z)=\log|B+C|$,
while the right-hand side is at most
$\log|A+B|+\log|A+C|$.
Exponentiating gives \eqref{eq:ruzsa}.
\end{proof}

\subsection{The Pl\"unnecke--Ruzsa inequality and its strengthening}
\label{sec:ruzsa-lemma}

We first recall Pl\"unnecke's inequality to explain the historical connections between the results below. Pl\"unnecke's inequality is the special case of Ruzsa' lemma below (Theorem~\ref{thm:ruzsa-lemma}).
\begin{theorem}[Pl\"unnecke's inequality \cite{plunnecke70}]
\label{thm:plunnecke}
For all finite nonempty $A,B\subseteq G$ and every $k \ge 1$,
\[
\mu(A,kB)
=\min_{\emptyset\ne S\subseteq A}\frac{|S+kB|}{|S|}
\le \Bigl(\frac{|A+B|}{|A|}\Bigr)^k,
\]
where $kB:=B+\cdots+B$ ($k$ times).
\end{theorem}

Ruzsa's sumset triangle inequality (Theorem~\ref{thm:ruzsa}) was
originally proved using different-summand Pl\"unnecke--Ruzsa inequality (Theorem~\ref{thm:ruzsa-lemma}),
which follows from Pl\"unnecke's inequality
(Theorem~\ref{thm:plunnecke}) and the tensor power trick. Later, Petridis \cite{petridis12} developed a lemma that directly implies both
the sumset triangle inequality and Pl\"unnecke's inequality. However, to our knowledge, no direct proof of Pl\"unnecke--Ruzsa inequality using
Petridis's machinery was previously known. This motivated us to seek an entropic counterpart of Pl\"unnecke--Ruzsa inequality, leading to Corollary~\ref{cor:k-fi-pr} and the strengthening in
Theorem~\ref{thm:ruzsa-lemma-strong}. Below, we give entropic proofs of both statements.

\begin{theorem}[Pl\"unnecke--Ruzsa inequality \cite{ruzsa09}]\label{thm:ruzsa-lemma}
For every $k \ge 1$ and all finite nonempty $A,B_1,\dots,B_k\subseteq G$,
\begin{equation}\label{eq:ruzsa-lemma}
\mu(A,B_1+\cdots+B_k) = \min_{\emptyset\ne S\subseteq A}
\frac{|S+B_1+\cdots+B_k|}{|S|}
\ \le\ \prod_{i=1}^k\frac{|A+B_i|}{|A|}.
\end{equation}
\end{theorem}

\begin{proof}
Write $D:=B_1+\cdots+B_k$, so the left-hand side is $\mu(A,D)$.
By Theorem~\ref{thm:magnif} and Corollary~\ref{cor:k-fi-pr},
\[
\begin{aligned}
\log\mu(A,D)
&=\min_{\substack{p_X\\ \supp p_X\subseteq A}}\Phi_D(p_X)
\le \min_{\substack{p_X\\ \supp p_X\subseteq A}}\sum_{i=1}^k\Phi_{B_i}(p_X)\\
&\le \sum_{i=1}^k\Phi_{B_i}(\unif_A)
\le \sum_{i=1}^k\log\frac{|A+B_i|}{|A|}.
\end{aligned}
\]
Exponentiating gives \eqref{eq:ruzsa-lemma}.
\end{proof}

\begin{theorem}[strengthening of Pl\"unnecke--Ruzsa inequality]\label{thm:ruzsa-lemma-strong}
For every $k \ge 1$ and all finite nonempty $A,B_1,\dots,B_k\subseteq G$,
\begin{equation}\label{eq:ruzsa-strong}
\mu(A,B_1+\cdots+B_k) = \min_{\emptyset\ne S\subseteq A}\frac{|S+B_1+\cdots+B_k|}{|S|}
\le
\min_{\emptyset\ne S\subseteq A}
\prod_{i=1}^k\frac{|S+B_i|}{|S|}.
\end{equation}
\end{theorem}

\begin{proof}
Write $D:=B_1+\cdots+B_k$. For every nonempty $S\subseteq A$,
evaluate the minimum in the proof of Theorem~\ref{thm:ruzsa-lemma}
at $\unif_S$ instead of $\unif_A$. Then
\[
\log\mu(A,D)
\le \sum_{i=1}^k\Phi_{B_i}(\unif_S)
\le \sum_{i=1}^k\log\frac{|S+B_i|}{|S|}.
\]
Exponentiating and minimizing over $S$ gives \eqref{eq:ruzsa-strong}.
\end{proof}



\subsection{A three-set sumset inequality}

\begin{proposition}\label{prop:abc}
Let $A,B,C\subseteq G$ be finite and nonempty. Define
\[
\cG(A,B+C)
:=
\left\{
p_{X,W}:
p_X=\unif_A,\ 
\supp p_W\subseteq B+C,\ 
p_{X+W}=\unif_{A+B+C}
\right\}.
\]
If $\cG(A,B+C)$ is nonempty, then
\begin{equation}\label{eq:abc}
|A||A+B+C|\ \le\ |A+B||A+C|.
\end{equation}
\end{proposition}

\begin{proof}
Choose $p_{X,W}\in\cG(A,B+C)$ and let $(X,W)\sim p_{X,W}$.
Choose a representation $W=Y+Z$, with $Y\in B$ and $Z\in C$,
as a function of $W$, and denote the induced marginals by $p_Y,p_Z$.
Then $X+Y+Z=X+W$ is uniform on $A+B+C$. By \eqref{eq:three},
\[
\begin{aligned}
\log|A|+\log|A+B+C|
&\le
\max_{\Pi(p_X,p_Y)}\HH(X+Y)
+\max_{\Pi(p_X,p_Z)}\HH(X+Z)\\
&\le \log|A+B|+\log|A+C|.
\end{aligned}
\]
The last inequality follows because the sums are supported on $A+B$ and $A+C$, respectively, under every corresponding coupling.
Exponentiating gives \eqref{eq:abc}.
\end{proof}


\begin{remark}[the hypothesis cannot be dropped]\label{rem:abc-false}
Inequality \eqref{eq:abc} is false in general. In $G=\ZZ_{20}$ take
\[
A=\{0,2,6,7,10,14,18\},\qquad B=\{11,15\},\qquad C=\{0,8\},
\]
so that $|A|=7$, $|A+B|=|A+C|=9$ and $|A+B+C|=13$. Then
$|A||A+B+C|=91>81=|A+B||A+C|$. Consistently, $\cG(A,B+C)$ is empty.
\end{remark}

\section{A non-abelian extension}\label{sec:nonab}

The direct analogue of Theorem~\ref{thm:three} is false in general
non-abelian groups. The inequality
\[
\HH(X)+\HH(X\circ Y\circ Z)\ \le\
\max_{\Pi(p_X,p_Y)}\HH(X\circ Y)+\max_{\Pi(p_X,p_Z)}\HH(X\circ Z)
\]
does not hold, and neither does $|A||B\circ C|\le|A\circ B||A\circ C|$. The
obstruction is the recombination step: preserving the distributions of
$(X,X\circ Z)$ and $(X\circ Z,X\circ Y\circ Z)$ need not preserve the distribution
of $Y$. With $U=Y\circ X\circ Z$, $R=Y\circ X$, and $S=X\circ Z$, however, both recombinations
preserve the prescribed marginals, and the argument survives.

\begin{theorem}[non-abelian max-coupling entropy inequality]
\label{thm:nonabelian-coupling}
Let $(G,\circ)$ be a group, let $A,B,C\subseteq G$ be finite and
nonempty, and let $(X,Y,Z)\sim p_{X,Y,Z}$ be any coupling of marginals
$p_X,p_Y,p_Z$ on $A,B,C$, respectively.
Then
\begin{equation} \label{eq:nonabelian-mc}
    \HH(X)+\HH_p(Y\circ X\circ Z)\ \le\
\max_{\Pi(p_X,p_Y)}\HH(Y\circ X)+\max_{\Pi(p_X,p_Z)}\HH(X\circ Z).
\end{equation}
\end{theorem}

\begin{proof}
Set $U=Y\circ X\circ Z$, $R=Y\circ X$, and $S=X\circ Z$,
and let $p_U$ be the induced distribution of $U$. Define
\[
\mathcal{M}(p_X,p_Y,p_Z,p_U) := \{q_{X,Y,Z}:\ q_X=p_X,\ q_Y=p_Y,\ q_Z=p_Z,\ q_U=p_U\}.
\]
As in the proof of Theorem~\ref{thm:three}, let $q^\star$
be the unique maximizer of $\HH_q(X,Y,Z)$ over $\mathcal M$.
The identities
\[
Y=R\circ X^{-1},\qquad Z=R^{-1}\circ U,
\qquad
Z=X^{-1}\circ S,\qquad Y=U\circ S^{-1}
\]
show that recombination on $(X,R,U)$ and on $(X,S,U)$
preserves all prescribed marginals. The same maximality argument
therefore gives
\[
\II_{q^\star}(X;U\mid R)
=\II_{q^\star}(X;U\mid S)=0.
\]
Since $U=R\circ X^{-1}\circ S$, any three of $X,R,S,U$
determine the fourth. The corresponding triples are also in
bijection with $(X,Y,Z)$, so the entropy identity in the proof of
Theorem~\ref{thm:three} applies:
\[
\HH_{q^\star}(R)+\HH_{q^\star}(S)
-\HH_{q^\star}(X)-\HH_{q^\star}(U)
=
\II_{q^\star}(R;S\mid X)+\II_{q^\star}(R;S\mid U)
\ge0.
\]
Using the prescribed marginals and bounding the two sum entropies
by their corresponding pairwise coupling maxima gives
\eqref{eq:nonabelian-mc}.
\end{proof}

\begin{corollary}[a non-abelian sumset inequality]
\label{cor:nonabelian-combinatorial}
For all finite nonempty subsets $A,B,C$ of a group $(G,\circ)$, we have
\begin{equation} \label{eq:nonabelian-sumset1}
    \sum_{x\in A}|B\circ x\circ C|
\ \le\ |B\circ A|\,|A\circ C|.
\end{equation}

Consequently,
\begin{equation} \label{eq:nonabelian-sumset2}
|A|\min_{x\in A}|B\circ x\circ C|
\ \le\ |B\circ A|\,|A\circ C|.
\end{equation}
\end{corollary}

\begin{proof}
For $x\in A$, put $n_x:=|B\circ x\circ C|$ and
$N:=\sum_{x\in A}n_x$.
Consider the distribution $p_X(x) = \frac{n_x}{N}$ defined over $A$. For every $x\in A$ and every
$u\in B\circ x\circ C$, fix a representation
\[
u=b_x(u)\circ x\circ c_x(u),
\qquad b_x(u)\in B,\quad c_x(u)\in C.
\]
Conditionally on $X=x$, let $U$ be uniform on $B\circ x\circ C$,
and set $Y=b_x(U)$ and $Z=c_x(U)$. Then
$U=Y\circ X\circ Z$ almost surely. The construction gives
a compatible coupling for Theorem~\ref{thm:nonabelian-coupling}, with 
\[
\HH(U\mid X=x)=\log n_x
\]
for every $x\in A$. Consequently,
\[
\begin{aligned}
\log N
&=\sum_{x\in A}\frac{n_x}{N}
\left(\log n_x-\log\frac{n_x}{N}\right)\\
&=\HH(U\mid X)+\HH(X)\\
&\le \HH(U)+\HH(X)\\
&\le
\max_{\Pi(p_X,p_Y)}\HH(Y\circ X)
+\max_{\Pi(p_X,p_Z)}\HH(X\circ Z)\\
&\le \log|B\circ A|+\log|A\circ C|.
\end{aligned}
\]
The penultimate inequality follows from
Theorem~\ref{thm:nonabelian-coupling}, and the last follows because the products are supported on $B\circ A$ and $A\circ C$, respectively.
Exponentiating gives \eqref{eq:nonabelian-sumset1}.
The minimum bound \eqref{eq:nonabelian-sumset2} follows by bounding each summand below by the minimum.
\end{proof}

\begin{remark}
This generalizes Ruzsa's sumset triangle inequality
(Theorem~\ref{thm:ruzsa}): in an abelian group, $|B+x+C|=|B+C|$ for
every $x\in A$, so \eqref{eq:nonabelian-sumset1} reduces to
$|A||B+C|\le |A+B||A+C|$.
\end{remark}

\section*{Conclusion}
In this work, we derive and extend some known inequalities in additive combinatorics using information-theoretic arguments. This adds to a growing body of work linking the two fields. 
\bibliographystyle{IEEEtran}
\bibliography{mylib}

\appendix
\numberwithin{theorem}{section}
\numberwithin{definition}{section}
\section{Alternative proof of Theorem~\ref{thm:magnif}} \label{apdx:magnif}

We give an alternative proof of Theorem ~\ref{thm:magnif} based on max-flow min-cut theorem from network flow theory. 

\begin{definition}[Flow network]
A flow network is a finite directed graph \(G=(V,E)\), together
with distinct vertices \(s,t\in V\), called the source and
sink, and a capacity function \(c:E\to\mathbb{R}_{\geq 0}\cup\{\infty\}\).
For \(A\subseteq V\), write
\[
\begin{aligned}
\delta^+(A)&=\{(u,v)\in E:u\in A,\ v\notin A\},\\
\delta^-(A)&=\{(u,v)\in E:u\notin A,\ v\in A\}.
\end{aligned}
\]
These are the edges leaving and entering \(A\), respectively.
For a single vertex \(v\), abbreviate
\(\delta^\pm(\{v\})\) by \(\delta^\pm(v)\).
\end{definition}

\begin{definition}[Flow and its value]
A feasible \(s\)--\(t\) flow is a function
\(f:E\to\mathbb{R}_{\geq 0}\) satisfying the capacity constraints
\[
0\leq f(e)\leq c(e)\qquad\text{for every }e\in E,
\]
and flow conservation at every vertex other than the source and sink:
\[
\sum_{e\in\delta^+(v)}f(e)
=\sum_{e\in\delta^-(v)}f(e)
\qquad\text{for every }v\in V\setminus\{s,t\}.
\]
The value of \(f\) is its outflow from the source:
\[
|f|
=\sum_{e\in\delta^+(s)}f(e)
=\sum_{e\in\delta^-(t)}f(e).
\]
The second equality follows from flow conservation.
\end{definition}

\begin{definition}[Cut and its capacity]
An \(s\)--\(t\) cut is a partition \((S,T)\) of \(V\) with
\(s\in S\) and \(t\in T\).
Its capacity $c(S,T)$ is
\[
c(S,T)=\sum_{e\in\delta^+(S)}c(e).
\]
\end{definition}
For every feasible flow \(f\) and every such cut $(S,T)$, conservation gives
\[
|f|
=\sum_{e\in\delta^+(S)}f(e) - \sum_{e \in \delta^-(S)} f(e)
\leq c(S,T).
\]

\begin{theorem}[max-flow min-cut theorem \cite{ff56}]
\label{thm:maxf-minc}
In every flow network $(G, s,t,c)$ as defined above, the maximum value of a feasible
\(s\)--\(t\) flow equals the minimum capacity of an \(s\)--\(t\) cut, if the latter is finite:
\[
\max_{\substack{f\in \mathcal F}} |f|
=\min_{\substack{S\subseteq V\\s\in S,\ t\notin S}}
c(S,V\setminus S),
\]
where $\mathcal F$ is the set of feasible $s$--$t$ flows. Both the maximum and the minimum are attained.
\end{theorem}

We use the following lemma.
\begin{lemma} \label{lem:magnif}
Let \(G=(A,B,E)\) be a nonempty bipartite graph without any isolated poiint in $A$, and magnification ratio
\[
\mu=\mu_{A\to B}(G):=\min_{\emptyset\neq S\subseteq A}\frac{|N(S)|}{|S|},
\]
where $N(S) \subseteq B$ is the
neighborhood of $S$. Let $\mathcal W(G)$ be the set of channels from $A$ to $B$
whose transition probabilities vanish outside $E$.
Then there exists a universal channel
$W_{Y\mid X}\in\mathcal W(G)$ such that, for every input
distribution $p_X$ on $A$,
\[
\HH(Y)-\HH(X)\ge\log\mu,
\]
where $p_Y=W_{Y\mid X}\circ p_X $ is the induced output distribution.
\end{lemma}

\begin{proof}
Introduce an auxiliary source $s$ and sink $t$, and construct the flow network \(s\to A\to B\to t\) with capacities
\[
c(s,x)=1,\qquad c(x,y)=+\infty\ \ ((x,y)\in E),\qquad c(y,t)=\frac1\mu.
\]
For any finite \(s\)-\(t\) cut, let \(A'\subseteq A\) be the vertices of \(A\) on the source side. To avoid an infinite-capacity edge crossing the cut, every neighbor of \(A'\) must also lie on the source side. Hence, the cut capacity is at least
\[
|A\setminus A'|+\frac{|N(A')|}{\mu}
\ge |A\setminus A'|+|A'|
=|A|,
\]
using \(|N(A')|\ge \mu|A'|\). The cut \(\{s\}\) has capacity \(|A|\), so the min-cut is \(|A|\). By Theorem~\ref{thm:maxf-minc}, there is a flow $f$ of value \(|A|\), where all edges \((s,x)\) are saturated. For $(x, y)\in A\times B$, define
\[
W_{Y\mid X}(y|x):=f(x,y).
\]
Then \(W_{Y\mid X}(y|x)=0\) for \((x,y)\notin E\), \(\sum_y W_{Y\mid X}(y|x)=1\), and
\[
\sum_x W_{Y\mid X}(y|x)\le \frac1\mu.
\]

Fix any $p_X$ on $A$. Let $u_A,u_B$ denote the uniform distributions on $A,B$, respectively, and put
$q_Y=W_{Y\mid X}\circ p_X$ and $r_Y=W_{Y\mid X}\circ u_A$.
By the data-processing inequality for relative entropy,
\[
D(q_Y\|r_Y)\le D(p_X\|u_A).
\]
Since \(\sum_x W(y|x)\le 1/\mu\),
\[
r_Y(y)\le \frac1{|A|\mu},
\]
so
\[
\frac{r_Y(y)}{u_B(y)}\le \frac{|B|}{|A|\mu}.
\]
Therefore
\[
\begin{aligned}
D(q_Y\|u_B)
&=D(q_Y\|r_Y)+\sum_y q_Y(y)\log\frac{r_Y(y)}{u_B(y)}\\
&\le D(p_X\|u_A)+\log\frac{|B|}{|A|\mu}.
\end{aligned}
\]
Using
\[
D(p_X\|u_A)=\log|A|-\HH(X),\qquad
D(q_Y\|u_B)=\log|B|-\HH(Y),
\]
we obtain
\[
\log|B|-\HH(Y)
\le \log|A|-\HH(X)+\log\frac{|B|}{|A|\mu}.
\]
Canceling terms gives
\[
\HH(Y)-\HH(X)\ge \log\mu.
\]
Thus the same channel \(W_{Y\mid X}\) works for every \(p_X\).
\end{proof}

\begin{proof}[Proof of Theorem~\ref{thm:magnif}]
     In a slight abuse of notation, to account for the additive nature of the settings considered earlier, we assume that the channel output is $X+Y$ and the input is $X$. Here $X$ takes values in $A$ and $Y$ takes values in $D$. Choose $S^*\in\arg\min_{\emptyset\ne S\subseteq A} |S+D|/|S|$ and put $p_X=\unif_{S^*}$. Then, for every channel $W_{Y|X}$ with output alphabet $D$, $\HH(X+Y)\le\log|S^*+D|$, and
    \[\min_{p_X} \max_{W_{Y|X}} \bigl[\HH(X+Y) - \HH(X)\bigr] \le \log \frac{|S^* + D|}{|S^*|} = \log \mu(A, D). \]

    For the reverse direction, setting $B = A + D$, and define the graph $G$ by adding edges between pairs $(a, a+d)$ for $a \in A$ and $d \in D$, Lemma~\ref{lem:magnif} asserts that for every $p_X$ supported on $A$,
\[\max_{W_{Y|X}} \bigl[\HH(X+Y) - \HH(X)\bigr] \geq \log \mu(A, D).\]
Taking the minimum over $p_X$ completes the proof.
\end{proof}
\begin{remark}
    While Prof. Venkat Anantharam was recently visiting CUHK, we discussed our research directions as well as the characterization of the magnification ratio using an entropic functional with him. He conveyed to us that Gemini gave him an alternate proof of the result using the max-flow min-cut theorem. While no further details were exchanged, we would still like to acknowledge this discussion.
\end{remark}

\section*{Bibliographic remarks and the use of AI}

These investigations originated as a research programme between Chandra Nair and Ken Lau (who has since graduated), and has now evolved into a joint research programme pursued by Chandra Nair, Aref Namayandeh, and Zhaobang Zhu. The original motivation was to employ the ideas from additive combinatorics to aid in establishing the Gaussian optimality for certain information functionals arising in the scalar Gaussian Z-interference channel. The core mathematical objectives, connections to Pl\"unnecke–Ruzsa theory, and the overarching strategy for translating information inequalities into additive-combinatorial statements were entirely human-conceived. 

Nair provided foundational prior work—including papers by Lau–Nair and Lau–Nair–Zhu—as context for an extended dialogue with an AI assistant (accessed via Perplexity). This assistant aided in exploring candidate proof strategies, deriving optimality conditions, implementing numerical experiments, and drafting working notes and manuscript revisions. As part of this iterative process, numerous proposed arguments failed, and several mathematical claims and numerical outputs required correction or rejection. In parallel, Namayandeh and Zhu conducted extensive independent numerical investigations into feasible couplings, generating evidence that propelled the search for a direct proof.

A separate AI system, suggested the Markov recombination construction, for a related version, used in the final proof. We explicitly acknowledge this substantive and elegant observation. Consequently, the final argument leverages recombination and maximality directly, bypassing the exploratory KKT calculations. The authors have written the entire article, with minimal stylistic help for polishing and notation from AI.

\end{document}